\documentclass[review]{elsarticle}

\usepackage{hyperref}

\usepackage{amsthm}		
\usepackage{amsmath}    
\usepackage{amssymb}

\usepackage{algorithmic}

\journal{Hal Archives-ouvertes.}

\newcommand{\Nullspace}{\bar{N}}

\newcommand{\rank}{\eta}
\newcommand{\nullity}{\bar\eta}
\newcommand{\substit}{\sigma}
\newcommand{\Null}{\varnothing}
\newcommand{\SAT}{{\tt{SAT}}}
\newcommand{\CNFSAT}{{\tt{cnf-SAT}}}
\newcommand{\KCNFSAT}{{\tt{k-cnf-SAT}}}
\newcommand{\CNFthreeSAT}{{\tt{3-cnf-SAT}}}

\newcommand{\onethreeSAT}{{\tt{1-3-SAT}}}
\newcommand{\onethreeSATplus}{{\tt{1-3-SAT$^{+}$}}}

\newcommand{\onekaySAT}{{\tt{1-K-SAT}}}

\newcommand{\ksubset}{\tt{0-1-IP$^=$}}

\newcommand{\Poly}{{\tt{P}}}
\newcommand{\NP}{{\tt{NP}}}
\newcommand{\coNP}{\tt{coNP}} 
\newcommand{\NPC}{{\tt{NP-complete}}}

\newcommand{\bigO}{\mathcal{O}}

\newcommand{\polyred}{\leq_{poly}}

\newcommand{\gauss}{\mathtt{GJE}}

\newcommand{\formula}{\varphi}
\newcommand{\formulas}{\Phi}

\newcommand{\unsat}{\bar\Sigma}
\newcommand{\sat}{\Sigma}

\newcommand{\degr}{\eta}
\newcommand{\freed}{\bar{\eta}}

\newcommand{\fract}{\kappa}

\newcommand{\K}{k}

\newcommand{\R}{r}

\newcommand{\system}{\mathtt{Sys}}

\newcommand{\Rationals}{\mathbb{Q}}
\newcommand{\Naturals}{\mathbb{N}}

\newtheorem{definition}{Definition}[section]
\newtheorem{theorem}{Theorem}[section]
\newtheorem{proposition}{Proposition}[section]

\newtheorem{remark}{Remark}[section]

\newtheorem{lemma}{Lemma}[section]

\newtheorem{corollary}{Corollary}[section]

\newtheorem{example}{Example}[section]

\begin{document}

\begin{frontmatter}

\title{
Positive 1-in-3-SAT admits a
non-trivial Kernel.
}
\author[mymainaddress]{Valentin Bura}

\address[mymainaddress]{
valentin.bura@gmail.com
}

\begin{abstract}
We illustrate the strength of
Algebraic Methods, adapting 
Gaussian Elimination and Substitution 
to the problem of
Exact Boolean Satisfiability.

\vfill 
\noindent 
For 1-in-3 SAT with non-negated literals
we are able to obtain 
considerably smaller equivalent 
instances of 0/1 Integer Programming
restricted to Equality.

\vfill
\noindent
Both Gaussian Elimination and 
Substitution may be used in a processing
step, followed by a type of brute-force
approach on the kernel thus obtained.

\vfill
\noindent
Our method shows that Positive instances
of 1-in-3 SAT 
may be reduced to significantly smaller
instances of I.P.E.
in the following sense. Any such instance
of $|V|$ variables and $|C|$ clauses
can be polynomial-time reduced to an
instance of 0/1 Integer Programming with
Equality, of size at most
$2/3|V|$ variables and at most 
$|C|$ clauses.

\vfill
\noindent
We obtain an
upper bound for the
complexity of counting,
$\bigO(2\fract\R 2^{(1-\fract)\R})$
for number of variables $r$ and
clauses to variables ratio $\fract$.

\vfill 
\noindent 
We proceed to define 
formally the notion of 
a non-trivial kernel,
defining the problems considered
as Constraint Satisfaction Problems. 

\vfill 
\noindent 
We conclude showing the 
methods presented here,
giving a non-trivial kernel for
positive 1-in-3 SAT, imply the existence
of a non-trivial kernel for
1-in-3 SAT.

\vfill

\newpage

\vfill
\noindent
Nous illustrons la force de
M\'{e}thodes alg\'{e}briques, adaptation
\'{E}limination et Substitution Gaussiennes
au probl\`{e}me de
Satisfabilit\'{e} bool\'{e}enne exacte.

\vfill
\noindent
Pour 1-sur-3-SAT avec des litt\'{e}raux 
non invers\'{e}s
nous pouvons obtenir
\'{e}quivalent consid\'{e}rablement plus petit
instances de programmation 
d'entiers 0/1
limit\'{e} \`{a} l'\'{e}galit\'{e} uniquement.

\vfill
\noindent
Élimination gaussienne et
La substitution peut être utilisée dans un traitement
étape, suivie d'un type de force brute
approche sur le noyau ainsi obtenu.

\vfill
\noindent
Notre méthode montre que les instances positives
de SAT 1-sur-3
peut être réduit à beaucoup plus petit
instances de I.P.E.
dans le sens suivant. 

\vfill
\noindent
Un tel cas
des variables $|V|$ et des clauses $|C|$
peut être réduit en temps polynomial à un
instance de Programmation Entier 0/1 avec
Égalité, de taille au plus
$2/3|V|$ variables et au plus
Clauses $|C|$.

\vfill
\noindent
Nous obtenons un
borne pour le
complexité du comptage,
$\bigO(2\fract\R 2^{(1-\fract)\R})$
pour le nombre de variables $r$ et
le rapport clauses/variables $\fract$.

\vfill
\noindent
Nous procédons à définir
formellement la notion de
un noyau non trivial,
définir les problèmes considérés
comme problèmes de satisfaction des contraintes.

\vfill
\noindent
Nous concluons en montrant le
méthodes présentées ici,
donner un noyau non trivial pour
SAT positif 1-sur-3, implique l'existence
d'un noyau non trivial pour
1-sur-3-SAT.

\vfill

\vfill\vfill
\vfill\vfill
\end{abstract}
\end{frontmatter}

\noindent
\textbf{Keywords.}
\textit{
Computational Complexity, Boolean Satisfiability,
Kernelization}.

\noindent
\textbf{
Mots clés.}
\textit{
Complexité informatique, Satisfaction booléenne,
Kernelization}.

\section{Introduction}

\noindent 
Recall that {\SAT}     
and its restrictions
{\CNFSAT},  
{\KCNFSAT}  and  {\CNFthreeSAT} 
are {\NPC} 
as shown in
\cite{cook71, karp72, levin73}.
The {\onethreeSAT} problem is 
that,
given a collection of triples over some variables,
to determine whether there exists a truth assignment 
to the variables so that each triple 
contains exactly one 
true literal and exactly two false literals. 

\vfill
\noindent
Schaefer's reduction given in \cite{schaefer78}
transforms
an instance of {\CNFthreeSAT} 
into a 
{\onethreeSAT} instance.
A simple truth-table argument shows
this reduction to be parsimonious,
hence {\onethreeSAT} is complete for the class
{\#\Poly} while 
a parsimonious reduction from 
{\onethreeSAT}
also shows {\onethreeSATplus} 
to be 
complete for {\#\Poly}.
We mention Toda's result 
in \cite{toda91} implying that 
P.H. is as hard computationally as
counting, which in a sense is our 
preoccupation here. 
A related result of
Valiant and Vazirani \cite{valiant85}
implies that detecting unique solutions
is as hard as NP. 
The algorithm we 
present does count solutions 
exhaustively, making use of preprocessing
and brute-force on the resulting kernel.

\vfill
\noindent
The {\onekaySAT} 
problem, a generalization
of {\onethreeSAT}, is that,
given a collection of 
tuples of size $K$ 
over some variables,
to determine whether there 
exists a truth assignment 
to the variables so that each 
$K$-tuple contains exactly one 
true 
and $K - 1$ false
literals.

\vfill
\noindent
The
{\onekaySAT} problem 
has been studied before
under the name of 
XSAT. 
In \cite{dahllof04} very strong 
upper bounds are given
for this problem, including the 
counting version.
These bounds are 
$\bigO(1.1907^{|V|})$ and 
$\bigO(1.2190^{|V|})$ respectively, while
in \cite{bjorklund08} the same bound of
$2^{|C|} |V|^{\bigO(1)}$ is given for both 
decision and counting, 
where $|V|$ is the number of
variables and $|C|$ the number of clauses.

\vfill
\noindent
Gaussian Elimination was used before in the context
of boolean satisfiability.
In \cite{soos10} the author uses this method for handling
xor types of constraints. 
Other recent examples of Gaussian elimination
used in exact algorithms or kernelization may be
indeed found in the literature
\cite{wahlstrom13, giannopoulou16}.

\newpage
\vfill
\noindent
Hence the idea that 
constraints of the type implying
this type of 
exclusivity can be formulated in terms
of equations, and therefore processed 
using Gaussian Elimination,
is not new and the intuition behind it is very
straightforward.

\vfill
\noindent
We mention the influential paper by 
Dell and Van Melkebeek
\cite{dell14} together with a continuation of their
study by Jansen and Pieterse \cite{jansen16,jansen17}.
It is shown in these papers that, under
the assumption that $\coNP\nsubseteq\NP\setminus \Poly$,
there cannot exist a significantly small kernelization 
of various problems, of which exact satisfiability is one. 
We shall use these results directly in our
current approach.

\vfill
\noindent
We begin our investigation by
showing how a {\onethreeSATplus}
instance can be turned into an
integer programming version
{\ksubset} instance with fewer 
variables. The number of variables in the
{\ksubset} instance is at most 
two-thirds of the
number of variables in the
{\onethreeSATplus} instance.
We achieve this by a straightforward
preprocessing of the 
{\onethreeSATplus} instance using 
Gauss-Jordan elimination.

\vfill
\noindent
We are then able to
count the solutions of
the {\onethreeSATplus} instance
by performing a brute-force
search on the {\ksubset}
instance.
This method 
gives interesting upper bounds 
on {\onethreeSATplus},
and the associated 
counting problem, though 
without a further analysis,
the bounds
thus obtained may not be the strongest
upper bounds found in the literature for 
these problems. 

\vfill
\noindent
Our method
shows how instances become easier
to solve with variation in 
clauses-to-variables ratio.
For random {\KCNFSAT} 
the ratio of clauses to variables
has been studied intensively, 
for example \cite{ding15} gives the proof
that a formula with density below a certain
threshold is with high probability satisfiable
while above the threshold is unsatisfiable. 

\vfill
\noindent
The ratio plays a similar role in our treatment
of {\onethreeSAT}.
Another important
observation is that in our case this ratio
cannot go below $1/3$ 
up to uniqueness of clauses,
at the expense of polynomial time pre-processing
of the problem instance.

\newpage
\vfill
\noindent
We note that, by reduction from {\CNFthreeSAT}
any instance of {\onethreeSAT} in which
the number of clauses does not exceed
the number of variables is 
{\NPC}. Hence we restrict our attention to these
instances.

\vfill
\noindent
Our preprocessing
induces a certain type of
``order" on the variables, such that
some of the non-satisfying 
assignments can be omitted by 
our solution search.
We therefore manage to dissect 
the {\onekaySAT} 
instance and obtain 
a ``core" of variables on which
the search can be performed.
For a 
treatment of Parameterized
Complexity 
the reader is directed  
to \cite{downey16}.

\vfill
\section{Outline}

\noindent
After a brief consideration of the notation used
in Section 3,
we define in Section 4 the 
problems 
{\onethreeSAT},
{\onethreeSATplus} and the associated counting
problems
{\#\onethreeSAT} and
{\#\onethreeSATplus}.

\vfill
\noindent
We elaborate on the relationship between the
number of clauses and the number of variables
in {\onethreeSATplus}.
We give a proof that {\onethreeSAT} is
$\NPC$ via reduction from {\CNFthreeSAT},
and a proof that {\onethreeSATplus}
is $\NPC$ via reduction from {\onethreeSAT}.

\vfill
\noindent
We conclude by remarking that due to this
chain of reductions, the restriction of 
{\onethreeSATplus} to instances with more 
variables than clauses is also $\NPC$,
since these kind of instances encode the
{\CNFthreeSAT} problem. We hence restrict 
our treatment of {\onethreeSATplus} to these
instances.

\vfill
\noindent
Section 5 presents our method of reducing   
a {\onethreeSATplus} instance to an instance
of 0/1 Integer Programming with Equality only.
This results in a 0/1 I.P.E. instance with
at most two thirds the number of variables
found in the {\onethreeSATplus} instance.

\newpage
\vfill
\noindent
Our method
describes the method sketched 
by Jansen and Pieterse in 
an introductory paragraph of \cite{jansen17}.
Jansen and Pieterse are not
primarily interested in reduction of the number of variables,
only in reduction of number of constraints. They do not 
tackle the associated counting problem.

\vfill
\noindent  
The method consists of encoding a {\onethreeSATplus}
instance into a system of linear equations and 
performing Gaussian Elimination on this system. 

\vfill
\noindent
Linear Algebraic
methods show the resulting matrix can be rearranged
into an $r\times r$ diagonal submatrix of ``independent'' 
columns, where $r$ is the rank of the system, 
to which it is appended a submatrix containing
the rest of the columns and the result column which
correspond roughly to the 0/1 I.P.E. instance
we have in mind.
We further know the values in the independent
submatrix can be scaled to $1$. 

\vfill
\noindent
The most pessimistic scenario
complexity-wise is when the input clauses, or the 
rank of the resulting system, 
is a third the number of variables, $|C| = 1/3|V|$, 
from which we obtain our complexity upper bounds.

\vfill
\noindent
To this case, one may wish to contrast the case of the 
system matrix being full rank, for which 
Gaussian Elimination alone suffices to find 
a solution.        
Further to this, we explain how to solve the 
0/1 I.P.E. problem in order to recover the number 
of solutions to the {\onethreeSATplus} problem.

\vfill
\noindent
Section 6 outlines the method of substitution,
well-known to be 
equivalent to Gaussian Elimination. 
Section 7
gives a worked example of this algorithm.
Section 8 and
Section 9 are concerned with an analysis
of the algorithm complexity, and 
correctness proof, respectively. 

\vfill
\noindent
Section 10 outlines the implications
for Computational Complexity, 
giving an argument that 
the existence of the {\onethreeSATplus}
kernel found in previous
sections implies the existence of a 
non-trivial kernel for the more general

\newpage
\vfill
\section{Notation}

\noindent
We write as usual
$A\polyred B$ to signify 
a polynomial-time reduction.

\vfill
\noindent
We denote boolean variables by
$p_1, p_2,\dots, p_i,\dots$
Denote the true and false constants by 
$\top$ and $\perp$ respectively.
For any SAT formula $\formula$,
write $\sat(\formula)$ if $\formula$
is satisfiable and write
$\unsat(\formula)$ otherwise. 
Reserve the notation $a(p)$ for
a truth assignment to the variable
$p$.

\vfill
\noindent
We write 
$\formulas(\R,\K)$
for the set of
formulas in $3$-CNF
with $\R$ variables
and $\K$ unique clauses.
We also write 
$\formula(V,C)$ to specify concretely
such a formula, where $V,C$
shall denote the sets of 
variables 
and clauses of $\formula$.
We write
$\fract(\formula) = \frac{\K}{\R}$.

\vfill
\noindent
We will make use of the following 
properties of a given map $f$:
\begin{list}{}{}
\item subadditivity:
$f(A + B)\leq f(A) + f(B)$ 

\item scalability:
$f(cA) = c f(A)$ for constant $c$.
\end{list}

\vfill
\noindent
For a given 
tuple $s = (s_1, s_2, \dots, s_n)$
we let $s(m)$ denote the element
$s_m$. 

\vfill
\noindent
For given 
linear constraints 
$L\colon\sum\limits_{i\leq n} d_i x_i = R$ 
for some $n$ and $x_i\in\{0,1\}$,
we let $L[x/L']$ be the result of substituting uniformly the expression
of constraint $L'$ for variable $x$ 
in constraint $L$. This is to be performed in 
restricted circumstances.

\vfill
\section{Exact Satisfiability}

\noindent
One-in-three satisfiability 
arose in late seventies as
an elaboration relating to 
Schaefer's Dichotomy Theorem
\cite{schaefer78}. 
It is proved there 
using certain assumptions 
boolean satisfiability problems 
are either in {\Poly} or they are
{\NPC}.

\vfill
\noindent
The
counting versions of 
satisfiability problems 
were introduced 
in \cite{valiant79} and it is known
in general that
counting is in some sense strictly harder 
than the corresponding decision
problem.

\newpage
\vfill
\noindent
This is due to the fact that,
for example, producing the number
of satisfying assignments of a formula
in $2$-CNF is complete for {\#\Poly},
while the corresponding decision
problem is 
known to be in {\Poly} \cite{valiant79}.
We thus restrict our attention
to {\onethreeSAT}
and more precisely
{\onethreeSATplus} formulas. 

\vfill
\begin{definition}[\onethreeSAT]
{\onethreeSAT} is defined as 
determining whether a formula 
$\formula\in\formulas(\R,\K)$ is satisfiable,
where the formula comprises of
a collection of triples
\[
C = 
\{\{p^{1}_{1}, p^{1}_{2}, p^{1}_{3}\},
\{p^{2}_{1}, p^{2}_{2}, p^{2}_{3}\},\dots,
\{p^{k}_{1}, p^{k}_{2}, p^{k}_{3}\}\}
\]
such that
$p_{1}^i,p_{2}^i,p_{3}^i\in 
V = 
\{p_1, \neg p_1, p_2, \neg p_2,\dots,
p_r, \neg p_r\}\cup\{\perp\}$ 
and
for any clause exactly one
of the literals is allowed 
to be true in an assignment, 
and no clause may contain
repeated literals or a literal
and its negation, and such that
every variable in $V$ appears in at
least one clause.

\vfill
\noindent
In the restricted case that
$p^{i}_{1},p^{i}_{2},p^{i}_{3}\in 
V^+ = 
\{p_1, p_2,\dots,
p_r\}\cup\{\perp\}$ 
for $1\leq i\leq r$
we denote the problem 
as \onethreeSATplus.

\vfill
\noindent
In the extended case that
we are required to 
produce the number of
satisfying
assignments, these problems 
will be denoted as
{\#\onethreeSAT} and 
{\#\onethreeSATplus}. 
\end{definition}

\vfill
\noindent

\begin{example}
The {\onethreeSATplus} formula
$\formula = \{\{p_1, p_2, p_3\}, 
\{p_2, p_3, p_4\}\}$ is satisfiable by
the assignment
$a(p_2) = \top$ 
and $a(p_j) = \perp$ for 
$j = 1,3,4$.
The {\onethreeSATplus} formula
$\formula = \{\{p_1, p_2, p_3\},
\{p_2, p_3, p_4\},
\{p_1, p_2, p_4\},
\{p_1, p_3, p_4\}\}$ 
is not satisfiable.
\end{example}

\vfill
\noindent
\begin{lemma}
Up to uniqueness of clauses 
and variable naming
the set
$\formulas(\R, \R/3)$
determines one
{\onethreeSATplus}
formula and this
formula is trivially satisfiable.
\end{lemma}
\begin{proof}
Consider the formula 
$\formula
=
\{\{p_{3i}, p_{3i+1}, p_{3i+2}\}
\mid
1\leq i\leq r/3
\}$ which has $r$ variables
and $r/3$ clauses, hence belongs 
to the set $\formulas(\R, \R/3)$ and
it is satisfiable, trivially, 
by any assignment
that makes each clause evaluate to true.

\vfill
\noindent
Now take any clause
$\{a,b,c\}\in
\formula(V,C)$ with 
$a,b,c \in V$.
We claim there is no
other
clause  
$\{a',b',c'\}
\in
\formula$ such that
$\{a,b,c\}
\cap
\{a',b',c'\}
\neq\emptyset$, for otherwise
let $a$ be in their 
intersection and we can  
see the number of variables
used by the $\R/3$ clauses
reduces by one variable, to be
$\R - 1$.
Now, since the clauses 
of $\formula$ do not overlap
in variables, we can see that
our uniqueness claim must hold,
since the elements of
$\formula$ are partitions of
the set of variables. 
\end{proof}

\vfill
\begin{remark}
For {\onethreeSATplus},
the sets
$\formulas(\R, \K)$
for $\K < \R/3$
are empty.
\end{remark}

\vfill
\noindent
Schaefer gives a polynomial 
time parsimonious reduction from
{\CNFthreeSAT} to {\onethreeSAT}
hence showing that {\onethreeSAT}
and its counting version
{\#\onethreeSAT}
are {\NPC} and
respectively \#\Poly{\tt{-complete}}.

\vfill
\noindent
\begin{proposition}[Schaefer, \cite{schaefer78}]
\label{13satnpc}
{\onethreeSAT} is {\NPC}. 
\end{proposition}
\begin{proof}
Proof by reduction from 
{\CNFthreeSAT}.
For a clause 
$p\vee p'\vee p''$
create three 
{\onethreeSAT} clauses
$\{\neg p, a, b\}$,
$\{p', b, c\}$,
$\{\neg p'', c, d\}$.
Hence, we obtain
an instance with 
$|V| + 4|C|$ variables
and $3|C|$ clauses,
for the instance of
{\CNFthreeSAT} of 
$|V|$ variables and
$|C|$ clauses.
\end{proof}

\vfill
\noindent
The following statement is given in 
\cite{garey02}.
For the sake of completeness,
we provide a proof by a parsimonious 
reduction 
from \onethreeSAT.

\vfill
\noindent
\begin{proposition}[Garey and Johnson\cite{garey02}]
\label{13satplusnpc}
{\onethreeSATplus} is {\NPC}.
\end{proposition}

\newpage
\vfill
\noindent
\begin{proof}
Construct instance of 
{\onethreeSATplus} from
an instance of {\onethreeSAT}.
Add every clause in the 
{\onethreeSAT} instance with 
no negation to the {\onethreeSATplus}
instance.

\vfill
\noindent
For every clause containing one negation
$\{\neg p, p', p''\}$,
add to the {\onethreeSATplus}
two clauses
$\{\hat{p}, p', p''\}$
and
$\{\hat{p}, p, \perp\}$ where 
$\hat{p}$ is a fresh variable.

\vfill
\noindent
For a clause containing
two negations 
$\{\neg p, \neg p', p''\}$
we add two fresh variables
$\hat{p}, \hat{p}'$
and three clauses
$\{\hat{p}, \hat{p}', p''\}$,
$\{\hat{p}, p, \perp\}$ and
$\{\hat{p}', p', \perp\}$.

\vfill
\noindent
For a clause containing
three negations 
$\{\neg p, \neg p', \neg p''\}$
we add three fresh variables
$\hat{p}, \hat{p}', \hat{p}''$
and four clauses
$\{\hat{p}, \hat{p}', \hat{p}''\}$,
$\{\hat{p}, p, \perp\}$ ,
$\{\hat{p}', p', \perp\}$ and
$\{\hat{p}'', p'', \perp\}$.

\vfill
\noindent
We 
obtain a {\onethreeSATplus} formula
with at most
$4|C|$ more clauses and at most
$|V| + 3|C|$ more variables, for
initial number of clauses and variables
$|C|$ and $|V|$ respectively.

\vfill
\noindent
Our reduction is parsimonious,
for it is verifiable
by truth-table
the number of satisfying assignments 
to the {\onethreeSAT} clause
$\{\neg p, \neg p', \neg p''\}$ 
is the same as the number of
satisfying assignments to the
{\onethreeSATplus} collection of 
clauses
$\{\{\hat{p}, \hat{p}', \hat{p}''\},
\{\hat{p}, p, \perp\},
\{\hat{p}', p', \perp\},
\{\hat{p}'', p'', \perp\}\}$. 
\end{proof}

\vfill
\noindent
\begin{remark}
In virtue of Theorem \ref{13satnpc}
and Theorem \ref{13satplusnpc}
we restrict ourselves to instances
of {\onethreeSATplus}
$\formula\in\formulas(\R,\K)$
with $\R\geq\K$.

\vfill
\noindent
For if an instance 
of {\CNFthreeSAT}
$\bar{\formula}\in\formulas(\R',\K')$
is reduced to an instance
of {\onethreeSAT}
$\hat{\formula}\in\formulas(\R'',\K'')$
then our reduction entails
$\R'' = \R' + 4\K'$ and 
$\K'' = 3\K'$.

\vfill
\noindent
We analyze 
the further reduction to the
instance of {\onethreeSATplus}
$\formula\in\formulas(\R,\K)$.
Let $C, C', C'', C'''$ 
be the collections of
clauses in $\hat{\formula}$
containing, no negation,
one negation, two negations
and three negations respectively.

\vfill
\noindent
Our reduction implies
$\R = \R'' + |C'| + 2|C''| + 3|C'''|$
and
$\K = |C| + 2|C'| + 3|C''| + 4|C'''|
= \K'' + |C'| + 2|C''| + 3|C'''|$.
Then,
$\R - \K = 
\R'' + |C'| + 2|C''| + 3|C'''| -
\K'' - |C'| - 2|C''| - 3|C'''|
 = 
\R'' - \K'' = 
\R' + 4\K' - 3\K' = \R' + \K' > 0
$.
\end{remark}

\newpage
\section{Rank of a Formula}

\vfill
\noindent
A rank function is used as a measure
of ``independence"
for members of a certain set. 
The dual notion of nullity is defined
as the complement of the rank.

\begin{definition}
A rank function $R$ obeys the following 

\begin{list}{}{}
\item 1. $R(\emptyset) = 0$,

\item 2. $R(A\cup B) \leq R(A) + R(B)$,

\item 3. $R(A) \leq R(A\cup\{a\}) \leq R(A) + 1$.
\end{list}
\end{definition}

\vfill
\noindent
\begin{definition}[rank and nullity]
For a {\onethreeSATplus}
formula $\formula(V,C)$ 
define the system of linear equations 
$\system(\formula)$ as follows:

\begin{list}{}{}
\item for any clause $\{p, p', p''\}\in C$
	add to $\system(\formula)$ equation
	$\bar{p}+ \bar{p}' + \bar{p}'' = 1$;
\end{list} 

\vfill
\noindent
Define the rank and nullity of $\formula$ as
$\degr(\formula) = R(\system(\formula))$
and
$\freed(\formula) = |V| - \degr(\formula)$.

\vfill
\noindent
If formula is clear from context, we also 
use the shorthand $\degr$ and $\freed$.
\end{definition}

\vfill
\noindent
\begin{remark}
$\degr$ is a rank function 
with respect to
sets of {\onethreeSAT} triples.
\end{remark}

\vfill
\noindent
\begin{lemma}
For any {\onethreeSATplus} instance
$\formula$
transformed into a linear system
$\system(\formula)$ we observe the
following:

\begin{list}{}{}
\item
$\bar{p} + \bar{p}' + \bar{p}'' = 1$
has a solution $S \subset \{0,1\}^3$ 
if and only if 
exactly one of 
$\bar{p}, \bar{p}', \bar{p}''$ 
is equal to $1$ 
and the other two are
equal to $0$.
\end{list}
\end{lemma}

\newpage
\vfill
\noindent
\begin{proposition}
For any formula $\formula\in\formulas(\R,\K)$
we have $\sat(\formula)$
if and only if
$\system(\formula)$
has at least one solution over
$\{0,1\}^{\R}$.
\end{proposition}

\vfill
\noindent
\begin{corollary}
\label{corsys}
A formula $\formula\in\formulas(\R,\K)$ has as many
satisfiability assignments as the
number of solutions of
$\system(\formula)$
over
$\{0,1\}^{\R}$.
\end{corollary}

\vfill
\noindent
We define the binary integer programming
problem with equality
here and show briefly 
that
{\onethreeSATplus} is reducible to 
a ``smaller'' instance of this 
problem. 

\vfill
\noindent
\begin{definition}[$0,1$-integer programming with equality]
The {\ksubset} problem is defined as
follows.
Given a family of finite tuples 
$s_1,s_2,\dots, s_k$
with each $s_i\in\Rationals^S$
for some fixed $S\in\Naturals$,
and given a sequence 
$q_1, q_2,\dots, q_k\in\Rationals$,
decide whether there exists 
a tuple 
$T\in\{0,1\}^S$ such that
\[\sum\limits_{j=1}^{S} s_i(j) T(j) = q_i
\text{ for each }
i\in\{1,2,\dots,k\}
\]  
\end{definition}

\vfill
\noindent
\begin{remark}
{\ksubset}$\in\bigO(k  2^S)$ 
where $k$ is the number of {\ksubset} tuples and
$S$ is the size of the tuples.  
\end{remark}

\vfill
\noindent
\begin{proof}
The bound is obtained 
through applying an exhaustive
search.
\end{proof}

\vfill
\noindent
\begin{lemma}
\label{degrfreed1}
Let $\formula\in\formulas(\R,\K)$
be a {\onethreeSATplus} formula, 
then
$\degr(\formula)
\leq\K$ and
$\freed(\formula)\geq 
\R - \K$.
\end{lemma}
\begin{proof}
Follows from
the observation that $\degr$ is 
a rank function.
\end{proof}

\newpage
\begin{lemma}
\label{lemmaalgor}
Consider a 
{\onethreeSATplus}  
formula
$\formula$ and suppose
$\degr(\formula) = k$ and
$\freed(\formula) = r - k$.
The satisfiability of $\formula$
is decidable in 
$\bigO(2k 2^{r-k})$.
\end{lemma}

\begin{proof}
The result of performing Gauss-Jordan Elimination on
$\system(\formula)$ yields, after
a suitable re-arrangement of column vectors, the
reduced echelon form

\[
\gauss(\system(\formula)) =
\begin{bmatrix}
1       & 0 & 0  & \dots   & 0 & x_{11} & x_{12} & \dots & x_{1d} & R_{1}\\
0       & 1 & 0  & \dots   & 0 & x_{21} & x_{22} & \dots & x_{2d} & R_{2}\\
0       & 0 & 1  & \dots   & 0 & x_{31} & x_{32} & \dots & x_{3d} & R_{3}\\
\vdots & \vdots & \vdots   & & \vdots & \vdots & \vdots &\vdots &\vdots &\vdots \\
0       & 0 & 0  & \dots   & 1 & x_{k1} & x_{k2} & \dots & x_{kd} & R_{k}\\
\end{bmatrix}
\]

\vfill
\noindent
Now
consider the following structure, obtained
from the given dependencies above through ignoring 
the zero entries 

\[
\begin{bmatrix}
1 & x_{11} & x_{12} & \dots & x_{1d} & R_1\\
1 & x_{21} & x_{22} & \dots & x_{2d} & R_2\\
1 & x_{31} & x_{32} & \dots & x_{3d} & R_3\\
\vdots & \vdots & \vdots & & \vdots & \vdots \\
1 & x_{k1} & x_{k2} & \dots & x_{kd} & R_k\\
\end{bmatrix}
\]

\vfill
\noindent
This induces an instance of {\ksubset}
which can be solved as follows
\begin{figure}[h]
\fbox{
\begin{minipage}{32em}
\small
{Initialize $C:= 0$}\\
{Enumerate all sequences $s\in S = \{0,1\}^d$}\\
{For each such sequence $s$:}\\
{If $\forall j\leq k
[\sum\limits_{i\leq d} s(i)x_{ji} = R_j 
\vee \sum\limits_{i\leq d} s(i)x_{ji} = R_j - 1]$ then 
$C\longleftarrow C+1$.}
\end{minipage}
}
\end{figure}

\vfill
\noindent
We note the length of sequences 
$s\in S$ is $d = r-k$, hence the brute force procedure
has to enumerate $2^{r-k}$ members of $S$.
Furthermore, each such sequence $s\in S$ is tested
twice against all of the constraints 
$x_{1i}, x_{2i},\dots, x_{ki}$ for $i\leq d$, resulting 
in the claimed time complexity 
of $\bigO(2k2^{r-k})$.

\newpage
\vfill
\noindent
To see the algorithm is correct, we give 
a proof that considers when the counter is incremented.  
Suppose for all $j\leq k$
some $s\in S$ is not a solution to either  
$\sum\limits_{i\leq d} s(i)x_{ji} = R_j$ or
$\sum\limits_{i\leq d} s(i)x_{ji} = R_j - 1$.
In this case, the counter is not incremented and 
we claim $s$ does not
induce a solution to the {\onethreeSATplus}
formula $\formula$.
For in this case $s$ is not a $0/1$ solution
to the system $\system(\formula)$ 
and hence by Corollary \ref{corsys} cannot be 
a satisfying solution to $\formula$. In effect,
the counter is not incremented as we have not seen
an additional satisfying solution.  

\vfill
\noindent
Now suppose for all $j\leq k$
some $s\in S$ is a solution to either  
$\sum\limits_{i\leq d} s(i)x_{ji} = R_j$ or
$\sum\limits_{i\leq d} s(i)x_{ji} = R_j - 1$.
In this case, the counter is incremented and 
we claim $s$ is indeed a solution to the
{\onethreeSATplus} formula $\formula$.

\vfill
\noindent
For if $s$ is a solution to all $j$th rows
constraint
$\sum\limits_{i\leq d} s(i)x_{ji} = R_j$
then $s$ satisfies the constraint
$x_{j1} + x_{j2} +\dots+ x_{jd} = R_j$
giving the satisfying assignment
$a(p) = \perp$ for all variables $p$ 
corresponding to variables in the 
diagonal matrix, and 
$a(p) = \perp$ for variables corresponding
to column $i$ for which $s(i) = 0$, and
$a(p) = \top$ for variables corresponding
to column $i$ for which $s(i) = 1$.

\vfill
\noindent
Similarly, if $s$ is a solution to all $j$th rows
constraint
$\sum\limits_{i\leq d} s(i)x_{ji} = R_j - 1$
then $s$ satisfies the constraint
$1 + x_{j1} + x_{j2} +\dots+ x_{jd} = R_j$
giving the satisfying assignment
$a(p) = \top$ if $p$ corresponds to 
the diagonal variable $(j,j)$,
$a(p) = \perp$ for all variables $p$ 
corresponding to all other variables in the 
diagonal matrix, and 
$a(p) = \perp$ for variables corresponding
to column $i$ for which $s(i) = 0$, and
$a(p) = \top$ for variables corresponding
to column $i$ for which $s(i) = 1$.  
\end{proof}

\vfill
\noindent
\begin{corollary}
\onethreeSATplus$\polyred$\ksubset.
\end{corollary}

\begin{proof}
By the pre-processing of the problem instance
using Gaussian Elimination, shown above,
{\onethreeSATplus} is reduced in polynomial
time to {\ksubset}.
\end{proof}

\newpage
\vfill
\noindent
\begin{theorem}
$\#$\onethreeSATplus$\in\bigO(\degr 2^{\freed + 1})$
for 
formula
rank and nullity 
$\degr$ 
and $\freed$.
\end{theorem}

\begin{proof}
There are $\degr$-many equations to satisfy by any 
assignment, and there are $\freed$-many variables
to search through exhaustively in order to solve
the {\ksubset} problem, which in turn solves the
{\onethreeSATplus} problem.
\end{proof}

\vfill
\noindent
\begin{corollary}
\label{mainresult}
$\#
\text{\onethreeSATplus}
\in \bigO(2\fract\R 2^{(1-\fract)\R})
$ for any instance
$\formula\in\formulas(\R,\K)$
and $\fract = \K/\R$.
\end{corollary}

\vfill
\section{The Method of Substitution}

\noindent
The substitution algorithm is
depicted below in Fig. 1.
We give a brief textual explanation
of the algorithm below.

\vfill
\noindent
\begin{list}{}{}

\item
\textbf{Pre-processing phase}

\item 1.
Let $n(c), m(c), s(c)$ be the lowest, middle and highest labeled 
variable in clause $c$. 
These values are distinct.

\item 2.
Represent clause $c$ in 
normal form as $n(c) = 1 - m(c) - s(c)$.

\item 3.
Sort the formula $\formula$
in ascending order of $n(c)$.
\end{list}

\vfill
\noindent
\begin{list}{}{}

\item 
\textbf{Substitution phase}

\item 1.
Initialize $i\leftarrow |C|$,

\item 2.
For each
$c_i\colon n(i) = a(i) - m(i) - s(i) = a(i) - M(i)$,

\item 3.
Initialize $j\leftarrow |C|$,

\item 4.
For each clause $c_j\in\formula$ with 
$j\neq i$
with
$c_j\colon n(j) = a(j) - m(j) - s(j)$
such that
$n(j)$ is found in the variables of
$m(i)$, or $n(j)$ is found in the variables of $s(i)$, do

\item 5.
Perform the substitution
$c_i\leftarrow c_i[n(j)/n(i)]$, 
and normalize the result.

\item 6.
Decrement variable $j$. Continue step 4.

\item 7.
Decrement variable $i$. Continue step 2.
\end{list}
\vfill

\newpage
\vfill
\noindent
\begin{figure}[h]
\label{algsubst}
\fbox{
\begin{minipage}{33em}
\begin{algorithmic}
\footnotesize
\STATE $i\gets k$
\STATE { }
\WHILE {$i\geq1$}
 \STATE { }
  \STATE $j\gets k$
   \STATE { }
	\WHILE {$j\geq 1$}
	 \STATE { } 
	  \IF {$j\neq i$} 
	   \STATE { }		       
		\IF {$c(j)\colon n(j)= a(j) - (\sum\limits_{x<r}d(x) + n(i))$,
		  		\textbf{and} $c(i)\colon n(i)= a(i) - (\sum\limits_{t<r}d(t))$} 
		    \STATE { }		
		     \STATE {$c(j)\leftarrow n(j) =  (a(j)-a(i)) - (\sum\limits_{x<r}d(x) + \sum\limits_{t<r}d(t))$}
		        \STATE { }	
		       \ENDIF  		
			  \STATE { }			      
		     \ENDIF		
			\STATE { }				    
		   \STATE $j\gets j-1$  		
		  \STATE { }	   	      
   	     \ENDWHILE
		\STATE { }	   	    
   	   \STATE $i\gets i-1$
   	  \STATE { }	 
     \ENDWHILE 
\end{algorithmic}
\end{minipage}}
\caption{Substitution algorithm}
\end{figure}

\vfill
\noindent
We remark an essentially cubic
halting time on the substitution
algorithm, which intuitively
corresponds to the cubic halting time of Gaussian Elimination, an equivalent method.
\vfill
\begin{remark}
Substitution halts in time
$\bigO(\K^2\times\R)$ for any formula
$\formula\in\formulas(\K,\R)$.
\end{remark}
\vfill

\newpage
\vfill
\noindent
Denote by $\sigma(\formula)$ or by $\sigma$,
when clear from context,
the structure thus obtained, denote by 
$\eta(\sigma)$ and $\bar{\eta}(\sigma)$
the rank and nullity thus induced, and
denote by 
$N(\sigma)$ and $\bar{N}(\sigma)$
the sets of independent, and dependent
variables generated through our process.

\vfill
\noindent
We remark the operator
$\substit$ is idempotent.
\begin{remark}
\label{substhalt}
$\sigma(\sigma(\formula))
=
\sigma(\formula)
$.
\end{remark}

\begin{proof}
Each clause $c(j)$ is read, and each
read clause is compared with every
other clause $c(i)$, in search for a 
common variable $n(i)$, if this variable
is found, a replacement is performed  
on $c(j)$.

\vfill
\noindent
Suppose
there exists a clause $c$ such that
$L = \substit(\formula)(c)
\neq
\substit(\substit(\formula))(c) = L'
$.  

\vfill
\noindent
Consider the case 
$L\setminus L'\neq\emptyset$.
Let variable 
$v$ be in this set difference.
It cannot be the case 
that $v = n(c)$ since this means
the procedure missed a mandatory 
substitution of $n(c)$, which the
second iteration picked up.

\vfill
\noindent 
Therefore $v\neq n(c)$. 
In this case, $v$ is the result of 
a chain
of substitutions ending with 
clause $c$.
An induction on this chain shows
the procedure missed a 
mandatory substitution of 
an $n$-variable, which the second
iteration picked up.

\vfill
\noindent
Consider the case 
$L'\setminus L\neq\emptyset$.
Let variable 
$v'$ be in this set difference.
It cannot be the case 
that $v' = n(c)$, since this means
the second iteration
introduced a new variable in
clause $c$, the result of a 
substitution of $n(c)$ which
the first iteration missed. 

\vfill
\noindent
Therefore $v\neq n(c)$. 
In this case, $v$ is a result of 
a chain
of substitutions ending with clause $c$.
An induction on this chain of 
substitutions shows the second iteration
of the procedure introduced a new variable
in clause $c$, the result of a 
substitution of an 
$n$-variable, which the first iteration
missed. 
\end{proof}

\newpage
\vfill
\noindent
As a consequence,
any set of formulas is closed under
substitution.
\begin{remark}
$\sigma[\sigma[\formulas]]
=
\sigma[\formulas]
$.
\end{remark}

\vfill
\section{An example}

\noindent
Consider the 
\onethreeSAT{ }
formula

\[
\formula = 
\{
\{
p_1, p_2, p_3
\},
\{
p_4, p_5, p_6
\},
\{
p_2, p_5, p_6
\},
\{
p_1, p_2, p5
\}
\}
\]

\vfill
\noindent
We outline the meaning of the rows and columns within our tabular format.
\vfill
\noindent
\begin{minipage}{32em}
  \begin{tabular}{| c | c | c | c |}
    \hline\hline
    $c_i$:& $n(i)$ & $m(i)$ & $c(i)$ \\ \hline\hline
  \end{tabular}
\quad
  \begin{tabular}{| c | c | c | c | c |}
    \hline\hline
    $C(i)$:& $n(i)$ & $a(i)$ & $m(i)$ & $c(i)$ \\ \hline\hline
  \end{tabular}
\end{minipage}

\vfill
\noindent
The formula $\formula$ 
is represented
in tabular format. 
Sort according to $n(i)$.

\vfill
\noindent
\begin{minipage}{32em}
  \begin{tabular}{| c | c | c | c |}
    \hline
    $c_1$:& 1 & 2 & 3 \\ \hline
    $c_2$:& 4 & 5 & 6 \\ \hline
    $c_3$:& 2 & 5 & 6 \\ \hline
    $c_4$:& 1 & 2 & 5 \\ \hline
    \hline
  \end{tabular}
  \quad
  \begin{tabular}{| c | c | c | c |}
    \hline
    $c_1$:& 1 & 2 & 3 \\ \hline
    $c_4$:& 1 & 2 & 5 \\ \hline
    $c_3$:& 2 & 5 & 6 \\ \hline
    $c_2$:& 4 & 5 & 6 \\ \hline
    \hline
  \end{tabular}
\end{minipage}

\vfill
\noindent
The formula is encoded as below.
Use a tabular data structure for the algorithm, initialized to empty.

\vfill
\noindent
\begin{minipage}{32em}
1.
  \begin{tabular}{| c | c | c | c | c |}
    \hline
    $C(1)$:& $\Null$ & $\Null$ & $\Null$ & $\Null$ \\ \hline
    $C(4)$:& $\Null$ & $\Null$ & $\Null$ & $\Null$ \\ \hline
    $C(3)$:& $\Null$ & $\Null$ & $\Null$ & $\Null$ \\ \hline
    $C(2)$:& $\Null$ & $\Null$ & $\Null$ & $\Null$ \\ \hline
    \hline
  \end{tabular}
\quad 
2. 
  \begin{tabular}{| c | c | c | c | c |}
    \hline
    $C(1)$:& $1$ & $0$ & $2$ & $3$ \\ \hline
    $C(4)$:& $1$ & $0$ & $2$ & $5$ \\ \hline
    $C(3)$:& $2$ & $0$ & $5$ & $6$ \\ \hline
    $C(2)$:& $4$ & $0$ & $5$ & $6$ \\ \hline
    \hline
  \end{tabular}
\end{minipage}

\newpage
\vfill
\noindent
Substitution phase. Operate on the data structure.
\vfill
\noindent
\begin{minipage}{32em}
3.
  \begin{tabular}{| c | c | c | c | c |}
    \hline
    $C(1)$:& $1$ & $1$ & $2$ & $3$ \\ \hline
    $C(4)$:& $1$ & $1-1$ & $5,6$ & $5$ \\ \hline
    $C(3)$:& $2$ & $1$ & $5$ & $6$ \\ \hline
    $C(2)$:& $4$ & $1$ & $5$ & $6$ \\ \hline
    \hline
  \end{tabular}
\quad
4.
  \begin{tabular}{| c | c | c | c | c |}
    \hline
    $C(1)$:& $1$ & $1-1$ & $5,6$ & $3$ \\ \hline
    $C(4)$:& $1$ & $1-1$ & $5,6$ & $5$ \\ \hline
    $C(3)$:& $2$ & $1$ & $5$ & $6$ \\ \hline
    $C(2)$:& $4$ & $1$ & $5$ & $6$ \\ \hline
    \hline
  \end{tabular}
\end{minipage}

\vfill
\noindent
Obtain the following partial result.
\vfill
\noindent
\begin{minipage}{32em}
5.
  \begin{tabular}{| c | c | c | c | c |}
    \hline
    $C(1)$:& $1$ & $0$ & $5,6$ & $3$ \\ \hline
    $C(4)$:& $1$ & $0$ & $5,6$ & $5$ \\ \hline
    $C(3)$:& $2$ & $1$ & $5$ & $6$ \\ \hline
    $C(2)$:& $4$ & $1$ & $5$ & $6$ \\ \hline
    \hline
  \end{tabular}
\end{minipage} 
  
\vfill
\noindent
Rearrange the tabular structure.

\vfill
\begin{tabular}{ c | c | c | c | c | c }
  \hline			
  $C_1$ :& $p_1$ & $=$ & $0$ & $-$ & $p_3, p_5, p_6$ \\
  $C_4$ :& $p_1$ & $=$ & $1$ & $-$ & $p_5, p_5, p_6$ \\
  $C_3$ :& $p_2$ & $=$ & $1$ & $-$ & $p_5, p_6$ \\
  $C_2$ :& $p_4$ & $=$ & $1$ & $-$ & $p_5, p_6$ \\
  \hline  
\end{tabular}

\vfill
\noindent
Note the result of the computation:

\[N = \{ p_1, p_2, p_4 \}
\text{ and }
\bar{N} = \{ p_3, p_5, p_6 \}
\] 

\vfill
\noindent
Hence the rank and nullity
of the formula are
$\degr(\formula) = 3$ and
$\freed(\formula) = 3$.

\vfill
\noindent
A Brute-Force Search on the set 
$\{p_3, p_5, p_6\}$
of dependent
variables 
yields the desired 
result 
to the 
\onethreeSATplus { } formula.

\vfill
\noindent
After the substitution process is finished, 
each of the clauses is expressed in terms of independent
variables, variables which cannot be expressed in terms
of other variables. We denote by 
$|n(i)|$ the number of variables in constraint $c(i)$
induced by the substitution method, excluding the variable
$n(i)$.

\newpage
\section{Algorithm Analysis}

\noindent
We maximize the number of
substitutions performed at each step.
Hence, at first step we encounter two substitutions,
at the second we encounter three substitutions,
while at every subsequent step we must assume
there exist two variables for which we can
substitute in terms of previously found
variables, which indicates that the formula 
for the Fibonacci expansion describes our process.

\vfill
\noindent
\begin{remark}
\label{remlargestexp}
The largest number of expansions determined by 
running substitution on the collection of clauses, is
\[|n(\K)| = 2, |n(\K-1)| = 3, 
|n(\K-2)| = 5,\dots,
|n(\K - i)| = Fib(i+3)\]  
\end{remark}

\vfill
\noindent
\begin{definition}[Representation]
The size of a
representation for a given
instance of 
\onethreeSATplus 
$\in\formulas(\R,\K)$
expressed by substitution
as
$n(1), n(2),\dots,n(\K)$
is
given by the formula
\[
r\times\log(\sum\limits_{i\leq\K}
|n(i)|)
\]
\end{definition}

\vfill
\noindent
\begin{remark}
The size of the resulting 
representation 
associated to formulas treated by 
Remark \ref{remlargestexp} 
converges asymptotically to
$
\R^2
\times 
\log(1.62)$.
\end{remark}
\begin{proof}
The bound is given by an analysis
of the growth of the Fibonacci
sequence.
It is well known the 
rate of growth of the
sequence 
converges approximately to
$1.62^n$. 
\end{proof}

\vfill
\noindent
\begin{remark}
\label{remnosubst}
Contrast the scenario in Remark \ref{remlargestexp},
to the case in which there are no substitutions
induced,
i.e.
$\formula = \{
\{p_{3i+1}, p_{3i+2}, p_{3i+3}\},
i\leq 1/3\K
\}$.
\end{remark}

\begin{remark}
The size of the resulting
representation 
associated to formulas treated by
Remark \ref{remnosubst} is 
$\R\times\log(2/3\R)$.
\end{remark}

\begin{proof}
In this case we have
$2\K$ independent variables, 
for a value of $\K$ of 
$1/3\R$. 
\end{proof}

\newpage
\vfill
\begin{theorem}
Any
{\onethreeSATplus} 
formula
admits a
representation 
with size 
$S$ for
\[
\R\times\log(2/3\R)
\leq
S
\leq
\R^2\times\log(1.62)
\]
\end{theorem}

\vfill
\begin{remark}
\label{repupperbound}
The size of any representation is
bounded above by
$\R^{2-\epsilon}$
for 
\[\epsilon 
= 
\frac{0.52}{\log(\R)}
\]
\end{remark}

\begin{proof}
$\R^{2-\epsilon}
=
\R^2
\times
\log(1.62)
$
implies
$2-\epsilon
=
2 + \frac{\log\log(1.62)}{\log(\R)}
$
and therefore

\[\epsilon
=
\frac{0.52}{\log(\R)}
\]
\end{proof}

\vfill
\section{Adequacy Proof}

\begin{proposition}
Let $\formula$ be a 
{\onethreeSATplus} formula
and let
$\substit = \substit(\formula)$
be the resulting structure
obtained by performing
substitution on $\formula$.
Then, 
$\rank(\formula)\leq 
\rank(\substit)$
and 
$\nullity(\formula)\geq 
\nullity(\substit)$.
\end{proposition}
\begin{proof}
It suffices to show that
$\nullity(\formula)\geq 
\nullity(\substit)$.

\vfill
\noindent
Suppose for a contradiction 
this is not the case.
We have that
$\nullity(\formula) <
\nullity(\substit)$. 

\vfill
\noindent
That is, that
the dependent variables of the
system of equations
exceed in number the dependent variables obtained through our
substitution algorithm.

\vfill
\noindent
We let 
$\nullity(\formula) =
\nullity(\substit) + K$. 
What this means is there exist
variables 
$p_1, p_2, \cdots, p_K$
such that
$p_i\in
\Nullspace(\substit)
\setminus
\Nullspace(\formula)$
for
$1\leq i\leq K$.

\vfill
\noindent
Take any such variable in this list
and perform another substitution
such as to decrease $K$ by one.
The existence of the list
$p_1, p_2, \cdots, p_K$ hence
contradicts the statement of
Remark \ref{substhalt}.
\end{proof}

\vfill
\newpage
\section{Implications}

\vfill
\noindent
\begin{proposition}[Schroeppel and Shamir\cite{schroeppel81}]
{\#\onethreeSAT}
can be solved in time $\bigO(2^{|V|/2})$ 
and space $\bigO(2^{|V|/4})$.
\end{proposition}

\vfill
\noindent
\begin{proposition}[Schroeppel and Shamir\cite{schroeppel81}]
{\#\ksubset}
can be solved in time 
$\bigO(2|C|2^{|V|/2})$ 
and space $\bigO(2^{|V|/4})$.
\end{proposition}

\vfill
\noindent
\begin{corollary}
{\#\onethreeSATplus}
can be solved in time
in time $\bigO(4/3|V|2^{3|V|/8})$ 
and space $\bigO(4/3|V|2^{3|V|/16})$.
\end{corollary}

\vfill
\noindent
Dell and Melkebeek \cite{dell14}
give a rigorous treatment of the  
concept of ``sparsification''.
In their framework,
an oracle communication protocol 
for a language $L$ is a communication protocol
between two players. 

\vfill
\noindent
The first player is 
given the input $x$ and is only allowed to run
in time polynomial in the length of $x$.
The second player is computationally unbounded,
without initial access to $x$.
At the end of communication, the first player 
should be able to decide membership in $L$.
The cost of the protocol is the length in bits
of the communication from the first player to the 
second.  

\vfill
\noindent
Therefore, if the first player is able to reduce, 
in polynomial time, the problem instance 
significantly, the cost of communicating
the ``kernel'' to the second player would also
decrease, hence providing us with a very natural
formal account for
the notion of sparsification.

\vfill
\noindent
Jansen and Pieterse in \cite{jansen16} 
state and give a procedure for
any instance of Exact Satisfiability
with unbounded clause length to be reduced
to an equivalent instance of the same
problem with only $|V| + 1$ clauses,
for number of variables $|V|$.

\newpage
\vfill
\noindent
The concern regarding the number of clauses
in {\onethreeSATplus} can be addressed, 
as we have done above. 
We observe that
for any instance $C$ of {\CNFthreeSAT},
the chain of polynomial-time parsimonious 
reductions
$C \to \hat{C} \to \bar{C}$, for
$\hat{C}$ and $\bar{C}$ instances of
{\onethreeSAT} and {\onethreeSATplus}
respectively, implies that 
the variables of 
$\hat{C}$ and $\bar{C}$
outnumber the clauses.

\vfill
\noindent
What is also claimed in \cite{jansen16}
is that, assuming
$\coNP\nsubseteq\NP\setminus \Poly$,
no polynomial time algorithm can 
in general transform
an instance of Exact Satisfiability
of $|V|$-many variables
to a significantly
smaller equivalent instance, 
i.e. an instance encoded
using $\bigO(|V|^{2-\epsilon})$ 
for any $\epsilon > 0$.

\vfill
\noindent
We believe it is already transparent
that, in fact, 
we have obtained a significantly 
smaller kernel for {\onethreeSATplus} 
above, i.e. transforming parsimoniously 
an instance of $|V|$ variables 
to a ``compressed'' instance of {\ksubset} 
of at most $2/3|V|$ variables.

\vfill
\begin{definition}[Constraint Satisfaction Problem]
A csp is a triple $(S,D,T)$ where
\begin{list}{-}{}
\item $S$ is a set of variables,
\item $D$ is the discrete domain the variables may range over,and
\item $T$ is a set of constraints.
\end{list}

\vfill
\noindent
Every constraint $c\in T$
is of the form $(t, R)$ where 
$t$ is a subset of $S$
and $R$ is a relation on $D$.
An evaluation of the variables is a function
$v\colon S\to D$. An evaluation $v$ satisfies 
a constraint $(t, R)$ if the values assigned to
elements of $t$ by $v$ satisfies relation $R$.
\end{definition}

\newpage
\vfill
\begin{remark}
The following are constraint satisfaction
problems: 

\begin{list}{-}{}

\item
{\CNFthreeSAT}

\item
{\onethreeSAT}

\item
{\onethreeSATplus}
\end{list}
\end{remark}

\vfill
\noindent
In what follows
we switch between notations and
write a csp in a more general form, 
with a problem
$(S,D,T)$ written as 
$L\subseteq\Naturals\times\Sigma^*$,
with instances $(k,x)$ 
such that $k = |S|$
and $x$ a string representation of $D$ and $T$.

\vfill
\begin{definition}[Kernelization]
Let $L,M$ be two parameterized 
decision problems,
i.e. $L, M\subseteq\Naturals\times\Sigma^*$
for some finite alphabet $\Sigma$. 

\vfill
\noindent
A kernelization for the problem 
$L$ parameterized by $k$ is 
a polynomial time reduction
of an instance $(k,x)$ to an instance
$(k',x')$ such that:
\begin{list}{-}{}
\item $(k,x)\in L$ if and only if $(k',x')\in M$,
\item $k'\in\bigO(k)$, and
\item $|x'| \in \bigO(|x|)$.
\end{list}
\end{definition}

\vfill
\noindent
\begin{definition}[Encoding]
An encoding of a problem 
$L\subseteq\Naturals\times\Sigma^*$
is a bijection
$h\colon L\to\Naturals$ such that
for any $(k,x)\in\Naturals\times\Sigma^*$
we have
$h(k,x)\in\bigO(|x|)$.
\end{definition}

\vfill
\noindent
\begin{definition}
A non-trivial kernel for 
{\CNFthreeSAT} is a 
kernelization of this problem transforming
any instance $\formula\in\formulas(\R,\K)$ 
to an instance $(f(\R),g(\K))$
of an arbitrary {\NPC} 
csp $M$, such that 
$f(\R)\in\bigO(\R)$ 
and $g(\K)\leq h(\K,\R)$ 
with $h(\K,\R)\in\bigO(\R^{3-\epsilon})$
for an encoding $h$ of $\formula$ and
some $\epsilon > 0$. 
\end{definition}

\newpage
\vfill
\noindent
\begin{remark}[Dell and Melkebeek \cite{dell14}]
\label{dellsremark}
{\CNFthreeSAT}
admits a trivial kernel 
$(f(\R), g(\K))$ with
$g(\K)\leq h(\K,\R)$ and
$h(\K,\R)\in\bigO(\R^3)$.
\end{remark}

\vfill
\noindent
\begin{lemma}[Dell and Melkebeek \cite{dell14}]
\label{dellslemma}
If {\CNFthreeSAT} admits a non-trivial kernel,
then $\coNP\subseteq\NP\setminus \Poly$.
\end{lemma}

\vfill
\noindent
\begin{definition}
A non-trivial kernel for 
{\onethreeSAT} is a 
kernelization of this problem
transforming any instance 
$\formula\in\formulas(\R,\K)$ 
to an instance 
$(f(\R),g(\K))$ 
of an arbitrary {\NPC} csp $M$,
such that  
$f(\R)\in\bigO(\R)$ 
and $g(\K)\leq h(\K,\R)$ 
with $h(\K,\R)\in\bigO(\R^{2-\epsilon})$
for
an encoding $h$ of $\formula$ and
some $\epsilon > 0$.
\end{definition}

\vfill
\noindent
\begin{remark}[Jansen and Pieterse \cite{jansen16}]
{\onethreeSAT} admits a kernel 
$(f(\R), g(\K))$ with
$g(\K)\leq h(\K,\R)$ and
$h(\K,\R)\in\bigO(\R^2)$.
\end{remark}

\vfill
\noindent
The following statement is given in 
\cite{jansen16}.
The authors elaborate on the results
of \cite{dell14} to analyze combinatorial
problems from the perspective of 
sparsification, and give several arguments
that non-trivial kernels for such problems
would entail a collapse of the Polynomial
Hierarchy to the level above 
$\Poly = \NP$.

\vfill
\noindent
It is essential to note here that this
line of reasoning was used by researchers
studying sparsification with the intention
of proving lower bounds on the existence of kernels, while the results presented by us
are slightly more optimistic.

\vfill
\begin{lemma}[Jansen and Pieterse \cite{jansen16}]
\label{lemJanPiet}
If {\onethreeSAT} admits a 
non-trivial kernel,
then $\coNP\subseteq\NP\setminus\Poly$.
\end{lemma}

\vfill
\noindent
\begin{lemma}
\label{lemXSAT}
If {\onethreeSATplus} admits a 
non-trivial kernel,
then 
{\onethreeSAT} admits a
non-trivial kernel.
\end{lemma}

\begin{proof}
Let $\formula\in\formulas(\R,\K)$ 
be an instance of 
{\onethreeSAT}.
By Schaeffer's results it follows
$\formula$ can be parsimoniously
polynomial time reduced to a 
{\onethreeSATplus} formula 
$\bar{\formula}\in\formulas(\R',\K')$
with $\R'= \R + 4\K$ and $\K'= 3\K$.

\newpage
\vfill
\noindent
Assuming {\onethreeSATplus} admits a
non-trivial kernel, 
this implies 
{\onethreeSAT} admits a non-trivial kernel,
and therefore through Lemma \ref{dellslemma}
$\coNP\subseteq\NP\setminus\Poly$. 

\vfill
\noindent
To spell this out, suppose we have
non-trivial kernel $(f(\R'),g(\K'))$ for
the problem {\onethreeSAT}, with
$g(\K')\leq h(\K',\R')$
and $h(\K',\R')\in\bigO(\R'^{2-\epsilon})$.
We observe
using the reduction from {\onethreeSAT}, 
$f(\R + 4\K)\leq 
f(\R) + 4f(\K)\leq 5f(\R)$ and therefore 
$f(\R')\in\bigO(\R)$
and, we obtain via the reduction the existence of 
a non-trivial kernel for 
{\onethreeSAT}, that is
$g(3\K)\leq 3g(\K)\leq 3h(\K,\R)$
with
$
h(\K,\R)\in\bigO(\R^{2-\epsilon})
$. 
\end{proof}

\vfill
\noindent
Essentially the following result is a 
restatement of Corollary \ref{mainresult}. 
\begin{theorem}
\label{thmain}
{\onethreeSATplus} admits a 
non-trivial kernel. 
\end{theorem}

\begin{proof}
Follows from
Lemma \ref{mainresult}.
The first player preprocesses 
the input in polynomial time using
Substitution, and passes the input
to the second player which makes use of
its unbounded resources to provide a solution
to this kernel.

\vfill
\noindent
It remains to show the cost of this 
computation is bounded non-trivially, i.e.
$h(\K,\R)\in\bigO(\R^{2-\epsilon})$
for $\epsilon > 0$.

\vfill
\noindent
This requirement
follows 
from Lemma \ref{mainresult}.
For the instance of {\ksubset} to which we reduce
has at most $f(\R')\leq 2/3\R$ variables and at most
$g(\K')\leq\R$ clauses.

\vfill
\noindent
We store the resulting instance
of {\ksubset} in a $(2/3\R + 1)\times\R$ matrix $M$
with polynomial-bounded entries, such that
$M(i,j) = d$ iff $d$  is the coefficient of
variable $i$ in constraint $j$, to which we add 
the result column.

\vfill
\noindent
From Remark \ref{repupperbound}
we obtain indeed that the bit representation
of this kernel is indeed $\R^{2-\epsilon}$
for some non-negative $\epsilon$.   
\end{proof}

\vfill
\begin{corollary}
$\coNP\subseteq\NP\setminus \Poly$
\end{corollary}

\begin{proof}
Follows from 
Lemma \ref{lemXSAT},
Theorem \ref{thmain} and
Lemma \ref{lemJanPiet}. 
\end{proof}

\newpage
\vfill
\section{Conclusion}

\vfill
\noindent
We have shown  
the mechanism through which
a {\onethreeSATplus}
instance can be 
transformed into an
integer programming version
{\ksubset} instance with
variables at most 
two-thirds of the
number of variables in the
{\onethreeSATplus} instance.

\vfill
\noindent
This was done by a straightforward
preprocessing of the 
{\onethreeSATplus} instance using 
the method of Substitution.

\vfill
\noindent
We manage to
count satisfying assignments to
the {\onethreeSATplus} instance
through a type of brute-force
search on the {\ksubset}
instance.

\vfill
\noindent
The method we have presented before 
in the shape of Gaussian Elimination
gives interesting upper bounds 
on {\onethreeSATplus}, and shows
how instances become harder to solve
with variations on the 
clauses-to-variables ratio. 

\vfill
\noindent
An essential
observation here 
is that in this case 
this ratio
cannot go below $1/3$ 
up to uniqueness of clauses.
This can be easily checked
in polynomial time..

\vfill
\noindent
By reduction from {\CNFthreeSAT}
any instance of {\onethreeSAT} in which
the number of clauses does not exceed
the number of variables is also 
{\NPC}.

\vfill
\noindent
Our contribution is in pointing out how the method of Substitution
together with a type of brute-force
approach suffice to find, constructively, a non-trivial kernel
for {\onethreeSATplus}.

\newpage
\vfill
\section*{Acknowledgments}

\vfill
\vfill
\noindent
Foremost thanks are due to Igor Potapov
for his support and benevolence shown towards this project. 

\vfill
\noindent
Most of the ideas presented here 
have crystallized while
the author was studying with Rod Downey 
at Victoria University of Wellington, 
in the New Zealand winter of 2010.

\vfill
\noindent
This work would have been much harder to write
without the kind hospitality of Gernot Salzer at TU Wien in 2013. 
There I have met and discussed with experts in the field 
such as Miki Hermann from 
École Politechnique. 

\vfill
\noindent
I was fortunate enough to attend
at TU Wien 
the outstanding exposition
in Computational Complexity 
delivered by 
Reinhard Pichler.  

\vfill
\noindent 
I am indebted to
Noam Greenberg for supervising my
Master of Science Dissertation in 2012. 

\vfill
\noindent
I thank Asher Kach, Dan Turetzky and
David Diamondstone for many useful thoughts
on Computability, Complexity and Model Theory.

\vfill
\noindent
I have also found useful
Dillon Mayhew's insights in Combinatorics, and 
Cristian Calude's research
on Algorithmic
Information Theory.

\vfill
\noindent
Exceptional logicians such as
Rob Goldblatt, Max Cresswell and
Ed Mares have also supervised various projects in which I was involved.

\vfill
\noindent
I further thank
Mark Reynolds and Tim French
from The University of Western Australia
for teaching me to think, and act under pressure.

\vfill
\noindent
Special acknowledgments are given to my colleague 
Reino Niskanen for useful comments and
proof reading an initial compressed version of this manuscript. 

\vfill\vfill\vfill
\hfill
Bucharest, June 2019
\vfill\vfill\vfill
\vfill\vfill\vfill
\vfill\vfill\vfill

\bibliography{mybibfile}
\vfill
\end{document}